\documentclass[11pt,letterpaper,oneside,english]{article}
\usepackage[top=3.2cm, bottom=3.2cm, left=3.2cm, right=3.2cm]{geometry}

\usepackage{mdwlist}

\usepackage{amssymb,amsbsy,latexsym}
\usepackage{amsmath}
\usepackage{graphics, subfigure, float}
\usepackage{fp, calc}
\usepackage{bm}

\usepackage[latin1]{inputenc}
\usepackage[american]{babel}
\usepackage[T1]{fontenc} 

\def\url#1{{\texttt #1}}

\usepackage{bm}

\usepackage{amscd,amsthm}
\usepackage{datetime}

\usepackage{pst-all}
\usepackage{verbatim, comment}

\newtheoremstyle{theorem}{1em}{1em}{\slshape}{0pt}{\bfseries}{.}{ }{}
\theoremstyle{theorem}
\newtheorem{theorem}{Theorem}
\newtheorem*{theorem*}{Theorem}
\newtheorem{corollary}[theorem]{Corollary}

\newtheorem{lemma}[theorem]{Lemma}

\theoremstyle{remark}

\newtheorem*{remark*}{Remark}

\providecommand{\setR}{\mathbb{R}}

\newcommand{\E}{\mathop{\mathbb{E}}}

 \theoremstyle{theorem}
\newtheorem*{homework*}{Homework}
 \theoremstyle{definition}
 
\newtheorem*{claim*}{Claim}

\newtheorem*{example*}{Beispiel}

        \def\drawRect#1#2#3#4#5{
           \FPeval{\x2}{(#2) + #4} 
           \FPeval{\y2}{(#3) + #5} 
           \pspolygon[#1](#2,#3)(\x2,#3)(\x2,\y2)(#2,\y2)
        }

\usepackage[displaymath,textmath,graphics, subfigure, floats]{preview} 
\PreviewEnvironment{myproblem}
\PreviewEnvironment{defn} 
\PreviewEnvironment{center} 
\PreviewEnvironment{pspicture} 
 \usepackage{datetime}

\makeatother

\begin{document}

\title{Constructive discrepancy minimization for convex sets}
\date{} 
\author{Thomas Rothvo{\ss}\thanks{Email: {\tt rothvoss@uw.edu}. The conference version of this paper
appeared at FOCS'14. Supported by NSF grant 1420180 with title ``\emph{Limitations of convex relaxations in combinatorial optimization}''.} 
\vspace{2mm} \\ University of Washington, Seattle} 
\maketitle

\begin{abstract}
A classical theorem of Spencer shows that any set system with $n$ sets and $n$ elements
admits a coloring of discrepancy $O(\sqrt{n})$. 
Recent exciting work of Bansal, Lovett and Meka shows that such colorings can be found in 
polynomial time. In fact, the Lovett-Meka algorithm finds a half integral point in 
any ``large enough'' polytope. However, their algorithm crucially relies on
the facet structure and does not apply to general convex sets. 

We show that for any symmetric convex set $K$ with Gaussian measure at least $e^{-n/500}$,
the following algorithm finds a point $y \in K \cap [-1,1]^n$ with $\Omega(n)$ coordinates in $\pm 1$:  
(1) take a random Gaussian vector $x$; (2) compute the point $y$ in $K \cap [-1,1]^n$ that is closest to $x$. (3) return $y$.

This provides another truly constructive proof of Spencer's theorem and the first constructive
proof of a Theorem of Gluskin and Giannopoulos. 
\end{abstract}

\section{Introduction}

Discrepancy theory deals with finding a \emph{bi-coloring} $\chi : \{ 1,\ldots,n\} \to \{ \pm 1\}$ of a set system $S_1,\ldots,S_m \subseteq \{ 1,\ldots,n\}$ 
so that the worst inbalance $\max_{i=1,\ldots,m} |\chi(S_i)|$ of a set
is minimized, where we denote $\chi(S_i) := \sum_{j \in S_i} \chi(j)$.
A seminal result of Spencer~\cite{SixStandardDeviationsSuffice-Spencer1985} 
says that there is always a coloring 
$\chi$ so that $|\chi(S_i)| \leq O(\sqrt{n})$ if $m = n$. The result is in particular
interesting since it beats the random coloring which has discrepancy $\Theta(\sqrt{n\log n})$.
Spencer's technique, which was first used by Beck in 1981~\cite{Beck-RothsEstimateIsSharp1981} is usually called
the \emph{partial coloring method} and is based on the argument that 
due to the pigeonhole principle many of the $2^n$ many colorings $\chi,\chi'$ must satisfy 
 $|\chi(S_i) - \chi'(S_i)| \leq O(\sqrt{n})$ for all sets $S_i$. Then one can take the \emph{difference}
between such a pair of colorings with $|\{ j \mid \chi(j) \neq \chi'(j) \}| \geq \frac{n}{2}$ to obtain a \emph{partial coloring} of low discrepancy.
Iterating the argument $\log n$ times provides a full coloring. 

Few years later and on the other side of the iron curtain, 
Gluskin~\cite{RootNDisc-Gluskin89} obtained the same result using convex geometry arguments.
In a paraphrased form, Gluskin's result showed the following:
\begin{theorem}[Gluskin~\cite{RootNDisc-Gluskin89}, Giannopoulos~\cite{Giannopoulos1997}]
For a small constant $\delta > 0$, let $K \subseteq \setR^n$ be a symmetric convex set 
with Gaussian measure $\gamma_n(K) \geq e^{-\delta n}$ and 
$v_1,\ldots,v_m \in \setR^n$ vectors of length $\|v_i\|_2 \leq \delta$. Then 
there are partial signs $y_1,\ldots,y_m \in \{ -1,0,1\}$ with $|\textrm{supp}(y)| \geq \frac{m}{2}$
so that $\sum_{i=1}^m y_iv_i \in 2K$.
\end{theorem}
For the proof, consider all $2^m$ many translates $\sum_{i=1}^m y_iv_i + K$ with $y \in \{ \pm 1\}^m$.
Then one can estimate that the total measure of the translates must be much bigger than 1, 
so there must be many pairs $y', y'' \in \{ \pm 1\}^m$
so that the translates overlap. Then take a pair that differs in at least half of the entries and
$y := \frac{1}{2}(y' - y'')$ gives the vector that we are looking for. 
For more details, we refer to the very readable exposition of Giannopoulos~\cite{Giannopoulos1997}.

In both, Spencer's original result and the convex geometry approach of Gluskin and Giannopoulos, 
the argument goes via the pigeonhole principle with exponentially many
 ``pigeons'' and ``pigeonholes'' which makes both type of proofs non-constructive. 
In a more recent breakthrough, Bansal~\cite{DiscrepancyMinimization-Bansal-FOCS2010} showed that a random walk, 
guided by the solution of an SDP 
can find the coloring for Spencer's Theorem in polynomial time. However, the approach needs a very careful
choice of parameters and the feasibility of the SDP still relies on the non-constructive argument. 
A simpler and truly constructive approach was provided by Lovett and Meka~\cite{DiscrepancyMinimization-LovettMekaFOCS12} who showed that
a ``large enough'' polytope of the form $P = \{ x \in \setR^n : \left|\left<v_i,x\right>\right| \leq \lambda_i \; \forall i \in [m] \}$ has a point $y \in P \cap [-1,1]^n$ 
that can be found in polynomial time and satisfies $y_i \in \{ -1,1\}^{n}$ for at least half of the coordinates. If the $v_i$'s are scaled to unit length, then the 
``largeness'' condition requires that  
\begin{equation} \label{eq:EntropyCondtion}
  \sum_{i=1}^m e^{-\lambda_i^2/16} \leq \frac{n}{16}.
\end{equation}
The approach of Lovett and Meka is surprisingly simple: start a random walk
at the origin and each time you hit one of the constraints $\left<v_i,x\right> = \pm \lambda_i$ or $x_i=\pm 1$, 
continue
the random walk in the subspace of the tight constraint. The end point of this random walk is the desired 
point $y$. 



Still, the algorithm of Lovett and Meka does not seem to generalize
to arbitrary convex sets and the condition in \eqref{eq:EntropyCondtion} might not be satisfied
for convex sets even if they have a large measure.  

\subsection{Related work}

If we have a set system $S_1,\ldots,S_m$ where each element lies in 
at most $t$ sets, then the partial coloring
technique from above can be used to  find a coloring of discrepancy $O(\sqrt{t} \cdot \log n)$~\cite{DiscrepancyBound-sqrtT-logN-SrinivasanSODA97}. A linear programming approach of
Beck and Fiala~\cite{IntegerMakingTheorems-BeckFiala81} shows that the discrepancy
is bounded by $2t-1$, independent of the size of the set system. 
On the other hand, there is a non-constructive approach of 
Banaszczyk~\cite{BalancingVectors-Banaszczyk98} that provides a bound of $O(\sqrt{t \log n})$ 
using a different type of convex geometry arguments. A conjecture of Beck and Fiala says that 
the correct bound should be $O(\sqrt{t})$. This bound can be achieved 
for the vector coloring version, see Nikolov~\cite{NikolovVectorColoringKomlosArxiv2013}.

More generally, the theorem of 
Banaszczyk~\cite{BalancingVectors-Banaszczyk98} shows that for any convex set $K$ with Gaussian measure at least $\frac{1}{2}$
and any set of vectors $v_1,\ldots,v_m$ of length $\|v_i\|_2 \leq \frac{1}{5}$, there exist 
signs $\varepsilon_i \in \{ \pm 1\}$ so that $\sum_{i=1}^m \varepsilon_iv_i \in K$.

A set of $k$ permutations on $n$ symbols induces a set system with $kn$ sets
given by the prefix intervals. One can use the partial coloring method
to find a  $O(\sqrt{k} \log n)$ discrepancy coloring~\cite{DiscrepancyOfPermutations-SpencerEtAl}, 
while a linear programming
approach gives a $O(k \log n)$ discrepancy~\cite{DiscrepancyOf3Permutations-Bohus90}. 
In fact, for any $k$ one can always color half of the elements with a
discrepancy of $O(\sqrt{k})$ --- this even holds for each induced sub-system~\cite{DiscrepancyOfPermutations-SpencerEtAl}.
Still, \cite{CounterexampleBecksPermutationConjecture-FOCS12} constructed  
3 permutations requiring a discrepancy of 
$\Theta(\log n)$ to color all elements. 

Also the recent proof of the Kadison-Singer conjecture by Marcus, Spielman
and Srivastava~\cite{KadisonSingerConjecture-MarcusSpielmanSrivastavaArxiv2013} can be 
seen as a discrepancy result. 
They show that a set of vectors $v_1,\ldots,v_m \in \setR^n$ with $\sum_{i=1}^m v_iv_i^T = I$
can be partitioned into two halfs $S_1,S_2$ so that $\sum_{i \in S_j} v_iv_i^T \preceq (\frac{1}{2} + O(\sqrt{\varepsilon}))I$ for $j \in \{ 1,2\}$
where $\varepsilon = \max_{i=1,\ldots,m}\{ \|v_i\|_2^2 \}$
and $I$ is the $n \times n$ identity matrix. Their method is based on
interlacing polynomials and no polynomial time algorithm is known to find
the desired partition.

For a very readable introduction into discrepancy theory, we recommend
Chapter~4 in the book of Matou{\v s}ek~\cite{GeometricDiscrepancy-Matousek99}
or the book of Chazelle~\cite{DiscrepancyMethod-Chazelle2001}.

\subsection{Our contribution}

Our main contribution is the following: 
\begin{theorem} \label{thm:MainTheorem}
There is a randomized polynomial time algorithm, which for any symmetric convex set $K \subseteq \setR^n$ with 
Gaussian measure at least $e^{-n/500}$ finds a point $y \in K \cap [-1,1]^n$ with $y_i \in \{ - 1,1\}$
for at least $\frac{n}{9000}$ many coordinates. Here it suffices if a polynomial time separation oracle
for the set $K$ exists.  
\end{theorem}

Our method is extremely simple: 
\begin{center}
\psframebox{
\begin{minipage}[b]{8.5cm}
{\bf Algorithm: } \vspace{1mm} \hrule \vspace{2mm}
\begin{enumerate}
\item[(1)] take a random Gaussian vector $x^* \sim N^n(0,1)$\vspace*{-2mm}
\item[(2)] compute the point \\ $y^* = \textrm{argmin}\{ \| x^* - y\|_2 \mid y \in K \cap [-1,1]^n\}$\vspace*{-2mm} 
\item[(3)] return $y^*$
\end{enumerate}
\end{minipage}
}
\psset{unit=1.5cm}
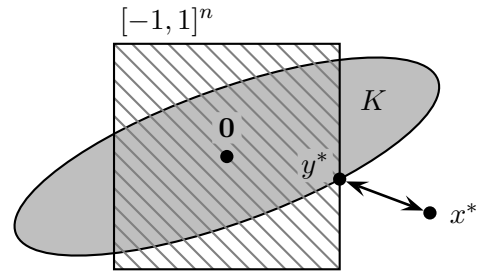
\begin{pspicture}(-2.4,-0.9)(1.5,1.2)
\rput[c]{20}(0,0){\psellipse[fillstyle=solid,fillcolor=lightgray](0,0)(2,0.6)}
\drawRect{fillstyle=vlines,fillcolor=lightgray,hatchcolor=gray}{-1}{-1}{2}{2}
\rput[c](1.3,0.5){\psframebox[framesep=2pt,fillstyle=solid,fillcolor=lightgray,linestyle=none]{$K$}}
\cnode*(0,0){2.5pt}{origin} \nput[labelsep=2pt]{90}{origin}{\psframebox[fillstyle=solid,fillcolor=lightgray,framesep=2pt,linestyle=none]{$\bm{0}$}}
\cnode*(1.8,-0.5){2.5pt}{x} \nput{0}{x}{$x^*$}
\cnode*(1,-0.2){2.5pt}{y} \nput[labelsep=0pt]{150}{y}{\psframebox[fillstyle=solid,fillcolor=lightgray,framesep=1pt,linestyle=none]{$y^*$}}
\rput[l](-0.95,1.2){$[-1,1]^n$}
\ncline[arrowsize=6pt,linewidth=1pt]{<->}{x}{y}
\end{pspicture}
\end{center}
\vspace{2mm}
In fact, the probability that the point $y^*$ satisfies the claim of Theorem~\ref{thm:MainTheorem}
is $1-2^{-\Omega(n)}$.

After the publication of the conference version of this paper, Eldan and Singh~\cite{EldanSinghArxiv14} 
discovered the following alternative algorithm: 
given a large enough symmetric convex body $K \subseteq \setR^n$, take a uniform random direction $c$ and optimize
the program $\max\{ cx \mid x \in K \cap [-1,1]^n\}$. The optimum solution $y$ will 
again have a constant fraction of coordinates in $\{-1,1\}$ with high probability.

\section{Preliminaries}

In the following, we write $x \sim N(0,1)$ if $x$ is a \emph{Gaussian random variable} with expectation $\E[x] = 0$ and 
variance $\E[x^2]=1$. By $N^n(0,1)$ we denote the \emph{$n$-dimensional Gauss distribution} and  
$\gamma_n$ denotes the corresponding measure with density $\frac{1}{(2\pi)^{n/2}} e^{-\|x\|_2^2/2}$ for $x \in \setR^n$. In other words, $\gamma_n(K) = \Pr_{x \sim N^n(0,1)}[x \in K]$ whenever $K$ is a measurable set. 
In fact, all sets $K$ that we deal with will be closed and convex and thus trivially measurable.

For a convex set $K$, let $d(x,K) := \min\{ \|x-y\|_2 \mid y \in K\}$ be the \emph{distance} of $x$
to $K$ and for $\delta \geq 0$, let $K_{\delta} := \{ x \in \setR^n \mid d(x,K) \leq \delta\}$
be the set of points that have at most distance $\delta$ to $K$ (in particular $K \subseteq K_{\delta}$).
A \emph{half-space} is a set of the form $H := \{ x \in \setR^n \mid \left<v,x\right> \leq \lambda \}$ for some
$v \in \setR^n$ and $\lambda \in \setR$.
The key theorem on Gaussian measure that we need is the \emph{Gaussian Isoperimetric inequality}
(see e.g. \cite{ProbabilityInBanachSpace-TalagrandLedoux91} for a proof):
\begin{theorem}
Let $K \subseteq \setR^n$ be a measurable set and $H$ be a halfspace so that $\gamma_n(K) = \gamma_n(H)$.
Then for any $\delta \geq 0$, $\gamma_n(K_{\delta}) \geq \gamma_n(H_{\delta})$.
\end{theorem}
A simple consequence is that any set $K$ that is not too small, is close to almost all the measure\footnote{Instead of using the
Gaussian isoperimetric inequality, one can prove Lemma~\ref{lem:DistanceFromBody} also using the well-known
\emph{measure concentration} inequality for Gaussian space: given a 1-Lipschitz function $F : \setR^n \to \setR$ 
(i.e. $|F(x) - F(y)| \leq \|x- y\|_2$) one has $\Pr_{x \sim N^n(0,1)}[|F(x) - \mu| > \lambda] \leq 2e^{-\lambda^2/2}$
with $\mu = \E_{x \sim N^n(0,1)}[F(x)]$. One can then choose $F(x) := d(x,K)$ with $\lambda := \frac{3}{2}\sqrt{\varepsilon n}$ 
and one obtains  $\Pr[|d(x,K) - \mu| > \frac{3}{2}\sqrt{\varepsilon n}] \leq 2e^{-\frac{9}{8}\varepsilon n} < e^{-\varepsilon n}$ for $n$ large enough. Since $\gamma_n(K) \geq e^{-\varepsilon n}$, 
we know that $\mu \leq \frac{3}{2}\sqrt{\varepsilon n}$ and thus $\Pr[d(x,K) > 2 \cdot \frac{3}{2}\sqrt{\varepsilon n}] \leq e^{-\varepsilon n}$ as claimed.}.
\begin{lemma} \label{lem:DistanceFromBody}
Let $\varepsilon > 0$. Then for any measurable set $K$
with $\gamma_n(K) \geq e^{-\varepsilon n}$ one has $\gamma_n(K_{3\sqrt{\varepsilon n}}) \geq 1 - e^{-\varepsilon n}$.
\end{lemma}
\begin{proof}
We assume that indeed $\gamma_n(K) = e^{-\varepsilon n} \leq \frac{1}{2}$.
Choose $\lambda \in \setR$ so that the halfspace $H = \{ x \in \setR^n \mid x_1 \leq \lambda\}$ has
measure $\gamma_n(H) = \gamma_n(K)$ (note that $\lambda \leq 0$). First, we claim that $|\lambda| \leq \frac{3}{2} \sqrt{\varepsilon n}$. 
This follows from 
\[
 \int_{-\infty}^{-\frac{3}{2} \sqrt{\varepsilon n}} \frac{1}{\sqrt{2\pi}} e^{-x^2/2}dx \leq 
e^{-\frac{9}{8}\varepsilon n}
\leq e^{-\varepsilon n}
\]
using the estimate 
$\int_t^{\infty} \frac{1}{\sqrt{2\pi}} e^{-x^2/2} dx \leq e^{-t^2/2}$ for all $t \geq 0$.
By symmetry, we get $\gamma_n(K_{3\sqrt{\varepsilon n}}) \geq 1 - e^{-\varepsilon n}$.
\end{proof}
For a vector $v \in \setR^n$ and $\lambda \geq 0$, the set $S = \{ x \in \setR^n : \left|\left<v,x\right>\right| \leq \lambda \}$ is called a \emph{strip}. If $v$ is a unit vector, then the strip has width $2\lambda$ and $\gamma_n(S) = \Phi(\lambda)$
where we define $\Phi(\lambda) := \int_{-\lambda}^{\lambda} \frac{1}{\sqrt{2\pi}} e^{-x^2/2}dx$. Useful estimates are  $\Phi(1) \geq e^{-1/2}$ and $\Phi(\lambda) \geq 1- e^{-\lambda^2/2}$ for all $\lambda \geq 0$. 

A convex body is called \emph{symmetric} if $x \in K \Leftrightarrow -x \in K$.
It is a convenient fact, that if we intersect a symmetric convex body with a strip, 
the measure decreases only slightly. 
\begin{lemma}[{\v{S}}id{\'a}k~\cite{SidaksLemma67}, Khatri~\cite{KhatriCorrelationInequality67}] \label{lem:SidakKhatri}
Let $K \subseteq \setR^n$ be a symmetric convex body and $S \subseteq \setR^n$ be a strip. 
Then $\gamma_n(K \cap S) \geq \gamma_n(K) \cdot \gamma_n(S)$.
\end{lemma}
The still unproven \emph{correlation conjecture} suggests that this claim is true 
for any pair $K,S$ of symmetric convex sets. 
For more details on Gaussian measures, see the book of Ledoux and Talagrand~\cite{ProbabilityInBanachSpace-TalagrandLedoux91}.

For $0\leq \varepsilon \leq 1$, let $h(\varepsilon) = \varepsilon \log_2(\frac{1}{\varepsilon}) + (1-\varepsilon) \log_2(\frac{1}{1-\varepsilon})$ be the 
\emph{binary entropy function}. Recall that for $0 \leq \varepsilon \leq \frac{1}{2}$,
the number of subsets $I \subseteq \{ 1,\ldots,n\}$ of size $|I| \leq \varepsilon n$ is bounded by\footnote{The argument is as follows: Define $\mathcal{S} := \{ S \subseteq [n] \mid |S| \leq \varepsilon n\}$ and let $X$ be the characteristic vector of a uniform random element from $\mathcal{S}$. If we define $H(X)$ as the \emph{entropy} of the random variable, then $\log_2( |\mathcal{S}|) = H(X) \leq \sum_{i=1}^n H(X_i) = \sum_{i=1}^n h(\Pr[X_i=1]) \leq n \cdot h(\varepsilon)$ using subadditivity of entropy as well as the monotonicity of $h$ on the interval $[0,\frac{1}{2}]$. }  $2^{h(\varepsilon) n}$.
One can easily estimate that $2^{h(\varepsilon)} \leq e^{\frac{3}{2} \varepsilon \log_2 (\frac{1}{\varepsilon})}$  which provides us with a 
 bound for later. 

A simple fact about convexity is that the optimum solution to a convex optimization problem does not change if we discard constraints that are not tight for the optimum.
Note that a function $g : \setR^n \to \setR$ is called \emph{strictly convex} if $g(\lambda x + (1-\lambda) y) < \lambda \cdot g(x) + (1-\lambda) \cdot g(y)$ for all $x,y \in \setR^n$ and $0<\lambda <1$.
\begin{lemma} \label{lem:ConvexOptimization}
Let $P,Q \subseteq \setR^n$ be convex sets and let $g : \setR^n \to \setR$
be a strictly convex function. Suppose that $x^*$ is an optimum solution to $\min\{ g(x) \mid x \in P \cap Q\}$
and $x^*$ lies in the interior of $Q$. Then $x^*$ is also an optimum solution to 
$\min\{ g(x) \mid x \in P\}$.
\end{lemma}
\begin{proof}
Suppose for the sake of contradiction that there is a $y^* \in P$ with $g(y^*) < g(x^*)$, 
then some convex combination $(1-\lambda)y^* + \lambda x^*$ with $0<\lambda<1$ lies also in $Q$
and has a better objective function than $x^*$, which is a contradiction.
\end{proof}

\section{Proof of the main theorem}

Now we have everything to analyze the algorithm. 
\begin{theorem} \label{thm:MainTechnicalTheorem}
Let $0 < \varepsilon \leq \frac{1}{9000}$ be a constant and $\delta := \frac{3}{2} \varepsilon \log_2(\frac{1}{\varepsilon})$. Suppose that $K \subseteq \setR^n$ is a symmetric, convex body with $\gamma_n(K) \geq e^{-\delta n}$. Choose a random Gaussian $x^* \sim N^n(0,1)$ and let $y^*$ be the point in $K \cap [-1,1]^n$ that
minimizes $\|x^* - y^*\|_2$. Then with probability $1-e^{-\Omega(n)}$, $y^*$ has at least $\varepsilon n$
many coordinates $i$ with $y_i^* \in \{ -1,1\}$.
\end{theorem}
\begin{proof}
First, we want to argue that $x^*$ has at least a distance of $\Omega(\sqrt{n})$
to the hypercube $[-1,1]^n$. A simple calculation shows that $\Pr_{x \sim N^n(0,1)}[|x_i| \geq 2] = 2\int_{2}^{\infty} \frac{1}{\sqrt{2\pi}}e^{-t^2/2}dt > \frac{1}{25}$. Then with probability
$1 - e^{-\Omega(n)}$ we have 
 $d(x^*,[-1,1]^n) \geq \sqrt{\frac{n}{25} \cdot (2-1)^2} = \frac{1}{5} \cdot \sqrt{n}$.

The crucial idea is that by the Gaussian isoperimetric inequality, $x^*$ will not 
be far from any  body that has a large enough Gaussian measure. 
The set $K \cap [-1,1]^n$ itself has only a tiny Gaussian measure, but we can instead consider 
the super-set 
$K(I^*) := K \cap \{ x \in \setR^n : |x_i| \leq 1 \; \forall i \in I^*\}$ where  $I^* := I^*(x^*) := \{ i \in [n] \mid y_i^* \in \{ \pm 1\} \}$
are the tight cube constraints for $y^*$.
We claim that $d(x^*,K \cap [-1,1]^n) = d(x^*,K(I^*))$ since the distance is 
already defined by the tight constraints for $y^*$! More formally, 
this claim follows from an  
application of Lemma~\ref{lem:ConvexOptimization} with $P := K(I^*)$, $Q := \{ x \in \setR^n \mid |x_i| \leq 1 \; \forall i \notin I^*\}$ and  $g(y) := \| x^* - y\|_2$ which is a strictly convex function.

Now, let us see what happens if $|I^*| \leq \varepsilon n$. 
We can apply the Lemma of {\v S}id{\'a}k and Khatri (Lemma~\ref{lem:SidakKhatri}) to lower bound the measure
of $K(I^*)$ as
\[
\gamma_n(K(I^*)) \geq \gamma_n(K) \cdot \prod_{i \in I^*} \gamma_n(\{ x \in \setR^n : |x_i| \leq 1\}) \geq \gamma_n(K) \cdot e^{-|I^*|/2} \geq e^{-\delta n} \cdot e^{-(\varepsilon/2)n} \geq e^{-2\delta n}
\]
using that strips of width 2 have measure at least $e^{-1/2}$ and that $\varepsilon \leq \delta$.
Now we know that the measure of  $K(I^*)$ is not too small and hence almost all Gaussian measure 
is close to it.  
Formally we obtain $\gamma_n(K(I^*)_{3\sqrt{2\delta n}}) \geq 1-e^{-2\delta n}$ by Lemma~\ref{lem:DistanceFromBody}. 
It seems we are almost done since we derived that with high probability, a random 
Gaussian vector has a distance of at most $3\sqrt{2\delta n}$ to $K(I^*)$ and
one can easily check that 
 $3\sqrt{2\delta n} < \frac{1}{5} \sqrt{n}$ 
for all $\varepsilon \leq \frac{1}{9000}$. But we need to be a bit careful since $I^*$
did depend on $x^*$. So, let us define $B := \bigcap_{|I| \leq \varepsilon n} (K(I)_{3\sqrt{2\delta n}})$. Observe that 
we have defined $\delta$ so that there are at most $e^{\delta n}$ many sets $I \subseteq [n]$ 
with $|I| \leq \varepsilon n$. Then by the 
union bound 
\[
\gamma_n(B) = 1-\gamma_n\Big( \bigcup_{|I| \leq \varepsilon n} (\setR^n \setminus K(I)_{3\sqrt{2\delta n}}) \Big) 
\geq 1-\sum_{|I| \leq \varepsilon n} \gamma_n(\setR^n \setminus K(I)_{3\sqrt{2\delta n}}) \geq 1 - e^{\delta n} \cdot e^{-2\delta n} \geq 1 - e^{-\delta n}.
\]
Now we can conclude that with probability $1-e^{-\Omega(n)}$, a random Gaussian
will have distance at least $\frac{1}{5}\sqrt{n}$ to the hypercube while at the same
time it has distance at most $3\sqrt{2\delta n} < \frac{1}{5} \sqrt{n}$ to all 
sets $K(I)$ with  $|I| \leq \varepsilon n$.
This shows that with high probability $|I^*| > \varepsilon n$.
\end{proof}
We get the constants as claimed in Theorem~\ref{thm:MainTheorem}
if we choose $\varepsilon = \frac{1}{9000}$ and observe that in this 
case $\delta \geq \frac{1}{500}$.


We should spend few words on the computational aspects of our algorithm.
We are assuming that for any point $x \notin K$, we can find a hyperplane separating $x$ 
from $K$ in polynomial time. First, $K$ must
be full-dimensional and even contain a ball of radius $r := e^{-\delta n}$ since otherwise $K$ would
be contained in a strip of width $2r$ which has a Gaussian measure of less than $r$. 
We can slightly modify the algorithm and output ``failure'' in case that $\|x^*\|_2 > R$ with $R := C \sqrt{n}$ for 
some some large enough constant $C$ --- this happens only with probability $e^{-\Omega(C^2)n}$. 
Now we can use the \emph{Ellipsoid method}~\cite{GLS-algorithm-Journal81} to find a point $\tilde{y}$ 
so that $\|\tilde{y} - y^*\|_2 \leq \eta$ where $\eta>0$ is some accuracy parameter. 
This can be done in time polynomial in $n$, $\log(\frac{1}{\eta})$, $\log(R)$ and $\log (\frac{1}{r})$. 
Now we can round that point $\tilde{y}$ to $\bar{y}$ with
\[
  \bar{y}_i := \begin{cases} 1 & \textrm{if } |\tilde{y}_i - 1| \leq \eta \\ -1 & \textrm{if } |\tilde{y}_i+1| \leq \eta \\ \tilde{y}_{i} & \textrm{otherwise} \end{cases}
\]
Then for $\eta < 1$ one has $y_i^* \in \{ -1,1\} \Rightarrow \bar{y}_i=y^*_i$. 
In particular $\bar{y} \in [-1,1]^n$ and the number of integral entries in $\bar{y}$ is at least $\varepsilon n$ as required.
Let $\|x\|_K := \min\{ \lambda \geq 0: x \in \lambda K\}$ denote the \emph{Minkowski norm} of $x$.
Then, $\bar{y}$ is almost in $K$ as 
\[
\|\bar{y}\|_K \leq \|y^*\|_K + \|\tilde{y} - y^*\|_K + \|\bar{y} - \tilde{y}\|_K \leq 1 + \frac{1}{r}\|\tilde{y}-y^*\|_2 + \frac{1}{r} \|\bar{y}-\tilde{y}\|_2 \leq 1+(n+1) \cdot \frac{\eta}{r}.
\]
Here we use that $\|z\|_K \leq \frac{\|z\|_2}{r}$ for all vectors $z \in \setR^n$ as $K$ contains a ball of radius $r$.
In order to actually obtain a point in $K$ one can apply the above algorithm to the slightly scaled 
body $K' := (1+(n+1) \cdot \frac{\eta}{r})^{-1}K$  and choose $\eta$ 
small enough so that $\gamma_n(K') \geq e^{-1.0001\delta n}$. The calculations in the proof of
Theorem~\ref{thm:MainTechnicalTheorem} have enough slack to account for the slightly reduced measure.

\section{Extension to intersection with subspaces}

As already mentioned, our algorithm includes the result of Lovett and Meka
in the following sense: Suppose our convex set is a polytope of the form
 $K = \{ x \in \setR^n : \left|\left< v_i,x \right>\right| \leq \lambda_i \; \forall i \in [m]\}$ where all the $v_i$'s are unit vectors
and $\lambda_i \geq 1$. In this case, the strip $S = \{ x \in \setR^n : \left|\left<v_i,x\right>\right| \leq \lambda_i\}$ of length $2\lambda_i$ has 
measure $\gamma_n(S) = \Phi(\lambda_i) \geq 1-e^{-\lambda_i^2/2} \geq \exp(-2e^{-\lambda_i^2/2})$ using that $\lambda_i \geq 1$. 
By the Lemma of {\v S}id{\'a}k-Khatri this means that
\[
  \gamma_n(K) \geq \prod_{i=1}^m \exp(-2e^{-\lambda_i^2/2}) = \exp\Big(-2\sum_{i=1}^m e^{-\lambda_i^2/2}\Big) \stackrel{!}{\geq} e^{-n/500}
\]
as long as $\sum_{i=1}^m e^{-\lambda_i^2/2} \leq \frac{n}{1000}$, exactly as in Lovett-Meka (apart from  different
constants). 
Please note that this line of arguments appeared already in 
the paper of Giannopoulos~\cite{Giannopoulos1997}. In the following we want to 
argue how $\Omega(n)$ many constraints with $\lambda_i=0$ can be incorporated 
in the analysis.

For a subspace $H$  we denote $N_H(0,1)$ as the 
$\dim(H)$-dimensional Gaussian distribution restricted to the subspace $H$
and we denote $\gamma_H$ as the corresponding measure. For example one can generate
a random $z \sim N_H(0,1)$ by selecting any orthonormal basis $u_1,\ldots,u_{\dim(H)}$ of $H$
and letting $z = \sum_{i=1}^{\dim(H)} g_iu_i$ where $g_1,\ldots,g_{\dim(H)} \sim N(0,1)$
are independent 1-dim. Gaussians. 
Note that $\gamma_H(H)=1$ and $\gamma_H(\setR^n \backslash H) = 0$. We want to remind the reader that for
any symmetric convex set $K$ and any subspace $H$, by log-concavity of $\gamma_n$ one has $\gamma_H(K) \geq \gamma_n(K)$.
More details can be found e.g. in Giannopoulos~\cite{Giannopoulos1997}.

We want to argue that the following variation of 
our main claim still holds:
\begin{theorem} \label{thm:FindYinKinsubspace}
Fix $0 < \varepsilon \leq \frac{1}{60000}$ and $\delta := \frac{3}{2} \varepsilon \log_2(\frac{1}{\varepsilon})$. 
Let $K \subseteq \setR^n$ be a symmetric, convex body 
with $K \subseteq H$ and $\gamma_H(K) \geq e^{-\delta n}$
where $H = \{ x \in \setR^n \mid \left<v_i,x\right>=0 \; \forall i\in [m] \}$
is a subspace defined by $m \leq 2\delta n$ equations.
Choose a random Gaussian $x^* \sim N^n(0,1)$ and let $y^*$ be the point in $K \cap [-1,1]^n$ that
minimizes $\|x^* - y^*\|_2$.
Then with probability $1-e^{-\Omega(n)}$, $y^*$ has at least $\varepsilon n$
many coordinates $i$ with $y_i^* \in \{ -1,1\}$.
\end{theorem}
\begin{proof}
Reinspecting the proof of Theorem~\ref{thm:MainTechnicalTheorem}, we see that it suffices to argue that 
most of the measure is still close to the sets $K(I)$. Formally, 
we will argue that for all $|I| \leq \varepsilon n$ one has
$\gamma_n(K(I)_{7\sqrt{2\delta n}}) \geq 1 - 2e^{-2\delta n}$. Then $7\sqrt{2\delta n} < \frac{1}{5}\sqrt{n}$ for $\varepsilon \leq \frac{1}{60000}$ and the claim follows. 

Hence, take a random point  $x^* \sim N^n(0,1)$ and let $z^* \in H$ be the projection of 
$x^*$ onto $H$ (that means $z^*$ is the point in $H$ closest to $x^*$). 
We may assume w.l.o.g. that $v_1,\ldots,v_m$ are orthonormal. 
First, at least some part of the measure is close to $H$, since
$\gamma_n(H_{\sqrt{2\delta n}}) \geq \gamma_n(\{ x \in \setR^n : \left|\left<v_i,x\right>\right| \leq 1 \; \forall i \in [m] \}) \geq e^{-2\delta n}$ by Lemma~\ref{lem:SidakKhatri}.
By Lemma~\ref{lem:DistanceFromBody} this implies that $\gamma_n(H_{4\sqrt{2\delta n}}) = \gamma_n( (H_{\sqrt{2\delta n}})_{3\sqrt{2\delta n}}) \geq 1 - e^{-2\delta n}$ 
and hence with the latter probability  $\|x^* - z^*\|_2 \leq 4\sqrt{2\delta n}$. 

In a second step, observe that we need to argue that $z^*$ is close to $K(I)$.
We know that $\gamma_H(K(I)) \geq \gamma_H(K) \cdot e^{-(\varepsilon/2) n} \geq e^{-2\delta n}$
as before. 
Since $z^*$ is an orthogonal projection of a Gaussian, we know that $z^* \sim N_H(0,1)$
and we obtain that $d(z^*,K(I)) \leq 3\sqrt{2\delta n}$ with probability $1-e^{-2\delta n}$. 
The claim then follows. 
\end{proof}

For being able to use the algorithm iteratively to find a full coloring, it is
important that we admit centers that are not the origin. But this is very straightforward to obtain.
In the following, for $c \in \setR^n$ and $K \subseteq \setR^n$ we define $c + K = \{ c+x : x \in K\}$
as the \emph{translate} of $K$ by $c$.
\begin{lemma} \label{lem:FindingPointInTranslatedSet}
Let $\varepsilon \leq \frac{1}{60000}$ and $\delta := \frac{3}{2}\varepsilon \log_2(\frac{1}{\varepsilon})$.
Given a subspace $H \subseteq \setR^n$ of dimension at least $(1-\delta)n$, 
a symmetric convex set $K \subseteq H$ with $\gamma_H(K) \geq e^{-\delta n}$ 
and a point $c \in {\left]-1,1\right[}^n$. There exists a polynomial time algorithm 
to find a point $y \in (c + K) \cap [-1,1]^n$ so that at least 
$\frac{\varepsilon}{2} n$ many indices $i$ have $y_i \in \{ -1,1\}$.
\end{lemma} 

\begin{proof}
For symmetry reasons we may assume that $0 \leq c_i<1$. Define a linear map $F : \setR^n \to \setR^n$
with $F((1-c_i) \cdot e_i) = e_i$, where $e_i$ is the $i$th unit vector. In other words, $F$
\emph{stretches} the space along the $i$th coordinate by a factor of $\frac{1}{1-c_i} \geq 1$. 
Note that in particular $F(\{x \in \setR^n : |x_i| \leq 1-c_i\}) = [-1,1]^n$.
Stretching can only increase the Gaussian measure, that means $\gamma_{F(H)}(F(K)) \geq \gamma_H(K)$ --- we will see formal arguments later in Cor.~\ref{cor:KscaledMultiDirections}. Moreover, $F(K)$ is still symmetric
and convex.
We can use Theorem~\ref{thm:FindYinKinsubspace} to find a vector $y \in F(K) \cap [-1,1]^n$ so that $|\{ i : y_i \in \{ -1,1\}\}| \geq \varepsilon n$. Again, after potentially replacing $y$ with $-y$ we may assume that $|\{ i : y_i = 1\}| \geq \frac{\varepsilon}{2} n$. We claim that the point $\tilde{y} := c+F^{-1}(y)$
will satisfy the claim. Since $y \in F(K)$, we have $F^{-1}(y) \in K$. 
Next, note that $\tilde{y}_i = c_i + (1-c_i) \cdot y_i$. Hence $\tilde{y} \in [-1,1]^n$ and for each $i$ with $y_i=1$ one has $\tilde{y}_i = 1$. This shows the claim.
\end{proof}
For the sake of completeness, we want to mention the slighly easier form of this lemma that
does not involve a subspace and has somewhat better constants: 
\begin{corollary} \label{cor:FindingPointInTranslatedSetSimple}
Let $\varepsilon \leq \frac{1}{9000}$ and $\delta := \frac{3}{2}\varepsilon \log_2(\frac{1}{\varepsilon})$.
Given a symmetric convex set $K \subseteq \setR^n$ with $\gamma_n(K) \geq e^{-\delta n}$ 
and a point $c \in {\left]-1,1\right[}^n$, there exists a polynomial time algorithm 
to find a point $y \in (c + K) \cap [-1,1]^n$ so that at least 
$\frac{\varepsilon}{2} n$ many indices $i$ have $y_i \in \{ -1,1\}$.
\end{corollary} 
\begin{proof}
Use the same proof as in Lemma~\ref{lem:FindingPointInTranslatedSet}, by apply directly 
Theorem~\ref{thm:MainTechnicalTheorem}.
\end{proof}

We want to briefly outline how one can iteratively apply Lemma~\ref{lem:FindingPointInTranslatedSet}
in order to find a full coloring
(similar arguments can be found in~\cite{Giannopoulos1997}). 
Intuitively, whenever we induce on a subset of coordinates, 
the convex set needs to be still large enough. 
For a subset $J \subseteq [n]$ of indices, we call $U = \{ x \in \setR^n : x_i = 0 \; \forall i \in J\}$
an \emph{axis-parallel subspace}.
\begin{lemma} \label{lem:FindingColoringWithLogNBlowup}
Suppose that $K \subseteq \setR^n$ is a symmetric convex body so that for all
axis-parallel subspaces $U \subseteq \setR^n$ 
one has that $\gamma_U(K) \geq e^{-\dim(U)/500}$. Then there is a polynomial time algorithm
to compute a $y \in \{ \pm 1\}^n \cap O(\log n) \cdot K$.
\end{lemma}
\begin{proof}
For iterations $t=1,\ldots,T$ we will compute a sequence of points $y^{(t)} \in (y^{(t-1)} + K) \cap [-1,1]^n$
that ends with the desired vector  $y := y^{(T)} \in \{ -1,1\}^n$.
We start with $y^{(0)} := \bm{0}$. Then in iteration $t \geq 1$, we define the subspace 
 $U := \{ x \in \setR^n : x_i = 0\;\mathrm{ for }\;  y^{(t-1)}_i \in \{ \pm 1 \} \}$ of variables
that have not been fixed so far. Then we apply
Cor.~\ref{cor:FindingPointInTranslatedSetSimple} with $\varepsilon := \frac{1}{9000}$ and $\delta \geq \frac{1}{500}$ to find a point $y^{(t)} \in y^{(t-1)} + (K \cap U)$. 
Note that in this application we consider $\setR^{\dim(U)}$ as the ambient space.
In each iteration a constant fraction of coordinates becomes integral and after
 $T = O(\log n)$ iterations we have $y^{(T)} \in \{ \pm 1\}^n$. 
We have $\|y^{(t)} - y^{(t-1)}\|_K \leq 1$ and hence
$\|y\|_K \leq T$ by the triangle inequality. This settles the claim.
\end{proof}
For Spencer's theorem it turns out that the $O(\log n)$-term can be replaced by $O(1)$ since the incurred
discrepancy bounds decrease from iteration to iteration. A general way to state this is as follows:

\begin{lemma}\label{lem:FindingColoringWithConstantBlowup}
Suppose that $K \subseteq \setR^n$ is a symmetric convex body so that
for all axis parallel subspaces $U \subseteq \setR^n$ one has
$\gamma_U((\frac{\textrm{dim}(U)}{n})^{\varepsilon} K) \geq e^{-\dim(U)/500}$ for some constant $\varepsilon > 0$. 
Then one can compute a vector $y \in \{ \pm 1\}^n \cap (c_{\varepsilon} K)$ in polynomial time.
\end{lemma}
\begin{proof}
Now we can apply the procedure from Lemma~\ref{lem:FindingColoringWithLogNBlowup}
even with a body $\tilde{K} := (\frac{\dim(U)}{n})^{\varepsilon} \cdot K$ that shrinks
over the course of the iterations. For some constant $0<c<1$ we have $\dim(U) \leq c^{t-1} \cdot n$
in iteration $t$, hence
$\|y\|_K \leq \sum_{t=1}^T \|y^{(t)}-y^{(t-1)}\|_K \leq \sum_{t=1}^{\infty} (\frac{c^{t-1}n}{n})^{\varepsilon} = \frac{1}{1-c^{\varepsilon}}$.
\end{proof}
Let us illustrate how to apply Lemma~\ref{lem:FindingColoringWithConstantBlowup}
in Spencer's setting. Consider a set system $S_1,\ldots,S_n \subseteq [n]$
with $n$ sets over $n$ elements and define a convex body  $K := \{ x \in \setR^n : |\sum_{j \in S_i} x_j| \leq 100\sqrt{n} \; \forall i \in [n]\}$.
If at some point we have already all elements except of $m$ many colored, then 
this means that we have a subspace $U$ of dimension $\dim(U) = m$ left. 
For such a set system with $m$ elements (but still $n \geq m$ sets), 
we can reduce the right hand side from $100\sqrt{n}$ to a value $100\sqrt{m \cdot \log \frac{2n}{m}}$  and the Gaussian measure is still large enough. 
More formally, if we want
$\gamma_U(\lambda \cdot K) \geq e^{-m/500}$, then a scalar of size
$\lambda = 100\sqrt{ m \cdot \log \frac{2n}{m} } /  (100\sqrt{n}) \leq (\frac{m}{n})^{1/5}$ suffices. 
Then Lemma~\ref{lem:FindingColoringWithConstantBlowup} finds a full coloring of discrepancy $O(\sqrt{n})$.

For the sake of completeness, we want to mention that after a modification of the
constants in \eqref{eq:EntropyCondtion}, the original argument of 
Lovett and Meka~\cite{DiscrepancyMinimization-LovettMekaFOCS12} could be adapted to provide $\frac{\varepsilon}{2}n$ integral coordinates 
while having $(1-\varepsilon)n$ many constraints $i$ with $\lambda_i = 0$.

\section{Extension to vector balancing}

The attentive reader might have realized that we have essentially proven Giannopolous' Theorem 
only in the variant in which the vectors $v_i$ correspond to the unit basis vectors. 
But we want to argue here that the algorithm from above can 
also handle Giannopoulos' general claim (apart from the fact that our partial
signs $x_i$ will be in $[-1,1]$ and not in $\{ -1,0,1\}$). 

For this sake, consider $Q = \{ x \in \setR^m \mid \sum_{i=1}^m x_iv_i \in K\}$. 
Then $Q$ is again a symmetric convex set and all we need to do is to
find a vector $y \in Q \cap [-1,1]^m$ that has $\Omega(m)$ many entries in $\pm 1$.
We know that it suffices to show that $\gamma_m(Q)$ is not too small --- and this is what
we are going to do now. 

First, let us discuss how the Gaussian measure of a body can change if 
we scale it in some direction: 
\begin{lemma} \label{lem:KscaledInOneDirection}
Let $K \in \setR^n$ be symmetric and convex and for some $\lambda \geq 0$ define
$Q := \{ (x_1,x_2,\ldots,x_n) \mid (\lambda x_1,x_2,\ldots,x_n) \in K\}$. Then $Q$ is symmetric and convex and
 $\gamma_n(Q) \geq \frac{1}{\max\{ 1, \lambda\} } \cdot \gamma_n(K)$.
\end{lemma}
\begin{proof}
Define $f(x_1) := \Pr_{x_2,\ldots,x_n \sim N(0,1)}[x \in K]$. 
Note that $f$ is a symmetric function and it is monotone in the 
sense that $0\leq x_1 \leq y_1 \Rightarrow f(x_1)\geq f(y_1)$. 
Then we can express both measures as
\begin{eqnarray*}
  \gamma_n(Q) &=& 2\int_{0}^{\infty} \frac{1}{\sqrt{2\pi}} e^{-x_1^2/2} \cdot f(\lambda x_1) \; dx_1 = 2\int_{0}^{\infty} \underbrace{\frac{1}{\sqrt{2\pi} \lambda} e^{-(x_1/\lambda)^2/2}}_{(*)} \cdot f(x_1)\; dx_1 \\ 
\gamma_n(K) &=& 2\int_{0}^{\infty} \underbrace{\frac{1}{\sqrt{2\pi}} e^{-x_1^2/2}}_{(**)} \cdot f(x_1) \; dx_1
\end{eqnarray*}
For $\lambda \leq 1$, we see that $f(\lambda x_1) \geq f(x_1)$ and hence $\gamma_n(Q) \geq \gamma_n(K)$. 
For $\lambda \geq 1$, we can estimate that 
$
\frac{(*)}{(**)} = \frac{1}{\lambda} \exp(\frac{1}{2} x_1^2 (1-\frac{1}{\lambda^2})) \geq \frac{1}{\lambda} 
$ and hence $\gamma_n(Q) \geq \frac{1}{\lambda} \gamma_n(K)$.
\end{proof}

Since also the scaled set $Q$ is symmetric, iteratively applying 
Lemma~\ref{lem:KscaledInOneDirection} gives:
\begin{corollary} \label{cor:KscaledMultiDirections}
Let $K \subseteq \setR^n$ be symmetric and convex and $\lambda \in \setR^n$. Then
\[
   \Pr_{x \sim N^n(0,1)}[(\lambda_1x_1,\ldots,\lambda_nx_n) \in K] \geq \frac{1}{\prod_{i=1}^n \max\{ 1,|\lambda_i|\}}  \Pr_{x \sim N_n(0,1)}[x \in K]
\]
\end{corollary}

\begin{lemma}
Let $v_1,\ldots,v_m \in \setR^n$ vectors with $\|v_i\|_2^2 \leq \beta$ for $i=1,\ldots,m$
and let $K \subseteq \setR^n$ be a symmetric convex set. 
For $Q = \{ x \in \setR^m \mid \sum_{i=1}^m x_iv_i \in K\}$
one has $\gamma_m(Q) \geq \gamma_n(K) \cdot e^{-\beta m}$.
\end{lemma}
\begin{proof}
We consider the random vector  $X = \sum_{i=1}^m x_iv_i$ with independent Gaussians $x_i \sim N(0,1)$. 
It is a well known fact in probability theory (see e.g. page 84 in \cite{IntroToProbability-FellerVol2-1971}), 
that there is an orthonormal basis  $b_1,\ldots,b_n \in \setR^n$ and $u \in \setR^n$ so that one can write $X = \sum_{i=1}^n y_iu_ib_i$
with $y_1,\ldots,y_n \sim N(0,1)$ being independent
Gaussians and the total variance of $X$ is preserved, that means $\|u\|_2^2 = \sum_{i=1}^m \|v_i\|_2^2$.
If we abbreviate $\Lambda :=  \prod_{i=1}^n \max\{ 1, |u_i| \}$, then 
we can apply Corollary~\ref{cor:KscaledMultiDirections} to lower bound
\[
\gamma_m(Q) = \Pr[X \in K] = \Pr_{y \sim N^n(0,1)}\Big[\sum_{i=1}^n y_iu_ib_i \in K\Big] \geq \frac{1}{\Lambda} \Pr_{y \sim N^n(0,1)}\Big[\sum_{i=1}^n y_ib_i \in K\Big] = \frac{1}{\Lambda} \gamma_n(K) 
\]
using the rotational symmetry of $\gamma_n$.
It remains to provide a (fairly crude) upper bound on $\Lambda$, which is
\[
  \Lambda  = \prod_{i=1}^n \max\{ 1, |u_i| \} 
\leq \prod_{i=1}^n (1+u_i^2) 
\stackrel{1+x \leq e^{x}}{\leq}
\exp\Big( \sum_{i=1}^n u_i^2\Big) = \exp\Big( \sum_{i=1}^m \|v_i\|_2^2 \Big) \leq e^{\beta m}
\]
\end{proof}

For example, if $\gamma_n(K) \geq e^{-m/1000}$ and $\|v_i\|_2^2 \leq \frac{1}{1000}$, then $\gamma_m(Q) \geq e^{-m/500}$
and we can apply Theorem~\ref{thm:MainTheorem} to obtain:
\begin{theorem}
Given a symmetric convex set $K \subseteq \setR^n$ with $\gamma_n(K) \geq e^{-m/1000}$ and vectors 
$v_1,\ldots,v_m \in \setR^n$, with $\|v_i\|_2 \leq \frac{1}{40}$ for all $i=1,\ldots,m$, 
there is a randomized polynomial time algorithm to find a $y \in [-1,1]^m$ with $\sum_{i=1}^m v_iy_i \in K$
and at least $\frac{m}{9000}$ many indices $i$ that have $y_i \in \{ \pm 1\}$.
Here it suffices to have access to a polynomial time separation oracle for $K$. 
\end{theorem}


\paragraph{Concluding remarks.} 

Finally, we want to repeat that it is still a wide open problem whether or not the
proof of Banaszczyk~\cite{BalancingVectors-Banaszczyk98} can be made constructive.

The author is very grateful to Daniel Dadush, Jakub Tarnawski and to the anonymous referees for their helpful comments.

\bibliographystyle{alpha}
\bibliography{constructive-discrepancy-minimization}

\end{document}